\documentclass[aps,pra,twocolumn,superscriptaddress,floatfix,
nofootinbib,showpacs,longbibliography]{revtex4-1}

\usepackage[utf8]{inputenc}  
\usepackage[T1]{fontenc}     
\usepackage[british]{babel}  
\usepackage[sc,osf]{mathpazo}
\usepackage[scaled=0.86]{berasans}  
\usepackage[colorlinks=true, citecolor=blue, urlcolor=blue]{hyperref}  
\usepackage{graphicx} 
\usepackage[babel]{microtype}  
\usepackage{amsmath,amssymb,amsthm,bm,amsfonts,mathrsfs,bbm} 

\usepackage{xspace}  
\usepackage{pgf,tikz}
\usepackage{xcolor}
\usepackage{multirow}
\usepackage{array}
\usepackage{bigstrut}
\usepackage{braket}
\usepackage{color}
\usepackage{natbib}
\usepackage{multirow}
\usepackage{mathtools}
\usepackage{float}
\usepackage[caption = false]{subfig}
\usepackage{xcolor,colortbl}
\usepackage{color}

\newcommand{\be}{\begin{equation}}
\newcommand{\ee}{\end{equation}}
\newcommand{\ba}{\begin{eqnarray}}
\newcommand{\ea}{\end{eqnarray}}

\newtheorem{theorem}{Theorem}

\newtheorem{definition}{Definition}
\newtheorem{proposition}{Proposition}

\newtheorem{lemma}{Lemma}

\begin{document}
	
\title{Genuinely entangled subspace with all-encompassing distillable entanglement across every bipartition}	

\author{Sristy Agrawal}
\affiliation{Department of Physical Sciences, Indian Institute of Science Education and Research Kolkata,	Mohanpur 741246, West Bengal, India.}

\author{Saronath Halder}
\affiliation{Department of Mathematics, Indian Institute of Science Education and Research Berhampur, Transit Campus, Government ITI, Berhampur 760010, India.}

\author{Manik Banik}
\affiliation{S.N. Bose National Center for Basic Sciences, Block JD, Sector III, Salt Lake, Kolkata 700098, India.}

\begin{abstract}
In a multipartite scenario quantum entanglement manifests its most dramatic form when the state is genuinely entangled. Such a state is more beneficial for information theoretic applications if it contains distillable entanglement in every bipartition. It is, therefore, of significant operational interest to identify  subspaces of multipartite quantum systems that contain such properties apriori. In this letter, we introduce the notion of unextendible biseparable bases (UBB) that provides an adequate method to construct genuinely entangled subspaces (GES). We provide explicit construction of two types of UBBs -- party {\it symmetric} and party {\it asymmetric} -- for every $3$-{\it qudit} quantum system, with local dimension $d\ge 3$. Further we show that the GES resulting from the symmetric construction is indeed a {\it bidistillable} subspace, i.e., all the states supported on it contain distillable entanglement across every bipartition. 
\end{abstract}



\maketitle
\section{Introduction} 
Entanglement is one of the most intriguing features of multipartite quantum system involving more than one spatially separated subsystems. States of a composite quantum system that are not  statistical mixture of product states of its constituent subsystems are called entangled \cite{Schrodinger35,Werner89}. This definition, which is basically a negation of a particular mathematical form, makes the identification, characterization, and quantification of entanglement a complicated task. The available entanglement detection schemes range from Bell inequalities \cite{Bell64}, EPR-Schr\"{o}dinger steering \cite{Einstein35,Schrodinger36}, entropic inequalities to the positive map based criteria (see \cite{Horodecki09, Guhne09, Kye13} and references therein). However, none of these methods are operational, i.e. for a given unknown state there is no known efficient algorithm to check whether the state is entangled or not. In fact, it has been shown that construction of universal entanglement witnesses with high accuracy is NP-Hard \cite{Gurvits04}.

To add to its complexity, entanglement is extremely fragile against decoherence \cite{Novotny11} and in most of the practical scenarios it appears in impure (mixed) form. But its pure form is known to be the most useful as operational resource. It is, therefore, desirable to obtain pure entanglement from the mixed ones through the process of entanglement purification/distillation \cite{Bennett96(1),Plenio07}. However, not all mixed entangled states are distillable. This leads us to the concept of bound entanglement -- entangled states with positive partial transpose (PPT) are the known such examples \cite{Horodecki97,Horodecki98}. It may be the case that these are not the only examples of bound entanglement as existence of such states with negative partial transpose (NPT) has been conjectured much earlier \cite{DiVincenzo00,Dur00}. Proving (or disproving) this conjecture is one of the hardest challenge in entanglement theory. In fact, it is not yet known whether question of distillability is at-all {\it decidable} \cite{Wolf11}. It is, therefore, utmost important to identify nontrivial subspaces (possibly of optimal dimension) of a bipartite Hilbert space with the feature that all density operators supported on those are distillable. The only example of such a subspace (of dimension $4$) arguably follows from the construction of (NPT) subspace of Ref.\cite{Johnston13} when combined with the recent result of Chen \& Djokovic \cite{Chen16}.

Study of entanglement becomes even more complicated for multipartite systems as different inequivalent types of entanglement arise there \cite{Horodecki09, Guhne09}. The most interesting form is when all the parties are genuinely entangled \cite{GHZ}. Several methods, albeit, not universal, are known to detect genuineness of the multipartite entanglement \cite{Guhne09}. On the other hand, due to different inequivalent classes of genuine entanglement \cite{Dur00(1)}, the notion of multipartite distillability does not have an obvious unique generalization of its bipartite counterpart. Nevertheless, distillability across bipartitions is still well defined. However, unlike the bipartite scenario, we do not know of any nontrivial multipartite subspace which is distillable across every bipartition.  

Given the difficulties in detecting and quantifying entangled states in general and more so for multipartite genuine entanglement, construction of subspaces possessing such properties apriori is a systematic way forward. Several algorithms for constructing completely entangled subspaces (CES) -- subspace containing no product vector -- are known for bipartite as well as multipartite scenarios \cite{Bennett99,Bennett99(1),Alon01,DiVincenzo03,Parthasarathy04,Hayden06,Niset06,Feng06,Cubitt08,Walgate08,Augusiak11,Childs13,Johnston14,Chen15,Halder18}. However, in a multipartite scenario, as entanglement appears in different inequivalent forms, a coherent formalism to construct subspaces containing only a particular type of entanglement is more demanding. In this letter, we introduce the concept of unextendible $k$-separable bases (UKB) which sufficiently serves the purpose. The case $k = 2$,  which we refer to as unextendible biseparable bases (UBB), evidently leads to a genuinely entangled subspace (GES)- subspace with the property that all density matrices supported on it are genuinely entangled. While a UBB can be constructed in different ways, our construction stems from multipartite unextendible product bases  (UPB). Here we note that, the unextendibility feature of a multipartite UPB is generally not preserved under different spatial configurations. Such a multipartite UPB, and to the best of our knowledge all the known examples, can be converted into a complete orthogonal bases by allowing entanglement among a subset of parties only. By assuring unextendibility across different spatial configurations one can obtain different inequivalent types of entangled subspaces, with GES being the most constrained one. We provide explicit construction of two types of UBBs - (i) party symmetric, (ii) party asymmetric, for every tripartite Hilbert space $(\mathbb{C}^d)^{\otimes 3}$, with $d\ge 3$. For the asymmetric construction, arguably, $d=3$ achieves the minimum local dimension required. Both these UBBs as well as the UPB from which they stem have elegant block structure(s) in a three dimensional {\it block-cube} of size $d\times d\times d$. We then proceed to analyze distillability feature of the constructed subspaces.  Interestingly, we show that the symmetric GES turns out to be bidistillable, in fact all the density matrices supported on it are $1$-distillable across every bipartition. However, this is not the case for asymmetric construction where the distillability feature appears only across certain partitions.  At this point it is important to note that the only other known construction of GESs has been proposed very recently in \cite{Demianowicz18}. Since this construction arises from nonorthogonal UPBs analysis of the distillability feature is not immediate there. We further show that the normalized projector onto the complementary subspace of the UPB is PPT entangled in every bipartition. Before proceeding to the main results we first briefly discuss some preliminary concepts.

\section{Notations and preliminary} 
An $m$-partite pure quantum state $\ket{\psi}_{k-sep}\in\otimes_{i=1}^{m}\mathbb{C}^{d_i}$ is called $k$-separable ($k\le m$), if and only if it can be written as a product of $k$ sub-states, i.e., $\ket{\psi}_{k-sep}=\otimes_{i=1}^{k}\ket{\psi}_{\mathcal{P}_i}$, where $\mathcal{P}_i$'s are the nonempty disjoint partitions of the party index, i.e., $\mathcal{P}_1\sqcup\cdots\sqcup\mathcal{P}_k=\{1,\cdots,m\}$. Convex mixtures of such states form $k$-separable density matrices. The cases $k=m$ and $k=2$ correspond to fully separable and biseparable states respectively. States that do not admit biseparable decomposition are genuinely entangled.

A set of pairwise orthogonal product vectors $\{\otimes_{j=1}^m|\psi\rangle^i_j\}_{i=1}^{n}$ spanning a proper subspace of $\otimes_{j=1}^m\mathbb{C}^{d_j}$ (i.e. $n<\Pi_{j=1}^md_j$) is called a UPB if its complementary subspace contains no product state \cite{Bennett99}. While the complementary subspace of a UPB is CES, our aim here is to construct a GES. To do so, we introduce 
a more generalized notion called UBB.
\begin{definition}\label{def1}
A set of pairwise orthogonal states $\{|\psi\rangle^i\}_{i=1}^n$ spanning a proper subspace of $\otimes_{j=1}^m\mathbb{C}^{d_j}$ is called an unextendible biseparable bases (UBB), if all the states $|\psi\rangle^i$ are biseparable and its complementary subspace contains no biseparable state. 
\end{definition}
The subspace complementary to the subspace spanned by the vectors in a UBB is a GES. Definition \ref{def1} can straightforwardly be generalized to introduce the notion of unextendible k-separable bases (UKB). Albeit slightly tedious, this, to the best of our knowledge, is the only prescription which helps construct and delineate inequivalent types of entangled subspaces in multipartite scenario. In present letter, however, our main focus is GES. Without further ado we now proceed to construct such GESs for tripartite systems. 
What follows next is the detailed constructions and analysis for three-qutrit system.

\section{Constructions in $(\mathbb{C}^3)^{\otimes 3}$} 
Consider the computational bases $\{|p\rangle_A\otimes|q\rangle_B\otimes|r\rangle_C~|~p,q,r=0,1,2\}$ for a three-qutrit system shared among Alice, Bob, and Charlie. The state can be represented in a $3\times3\times 3$ block-cube composed of $27$ basic blocks (analogous to {\it Rubik's cube}). The basic block indexed $`pqr'$ contains the state $\ket{p}_A\otimes\ket{q}_B\otimes\ket{r}_C$. Instead of the computational bases, let us now consider the following {\it twisted} orthogonal product bases (t-OPB), 
\begin{subequations}\label{tCOPB}
	\begin{align}
	\mathcal{B}_0:&=\{\ket{\psi}_{kkk}\equiv\ket{k}_A\otimes\ket{k}_B\otimes\ket{k}_C|k\in\{0,1,2\}\},\\
	\mathcal{B}_1:&=\{\ket{\psi(i,j)}_{1} \equiv \ket{0}_A\otimes\ket{\eta_i}_B\otimes\ket{\xi_j}_C\},\\
	\mathcal{B}_2:&=\{\ket{\psi(i,j)}_{2} \equiv \ket{\eta_i}_A\otimes\ket{2}_B\otimes\ket{\xi_j}_C\},\\
	\mathcal{B}_3:&=\{\ket{\psi(i,j)}_{3} \equiv \ket{2}_A\otimes\ket{\xi_j}_B\otimes\ket{\eta_i}_C\},\\
	\mathcal{B}_4:&=\{\ket{\psi(i,j)}_{4} \equiv \ket{\eta_i}_A\otimes\ket{\xi_j}_B\otimes\ket{0}_C\},\\
	\mathcal{B}_5:&=\{\ket{\psi(i,j)}_{5} \equiv \ket{\xi_j}_A\otimes\ket{0}_B\otimes\ket{\eta_i}_C\},\\
	\mathcal{B}_6:&=\{\ket{\psi(i,j)}_{6} \equiv \ket{\xi_j}_A\otimes\ket{\eta_i}_B\otimes\ket{2}_C\},
	\end{align}
\end{subequations}
where $i,j\in\{0,1\}$, and $\ket{\eta_i}:=\ket{0}+(-1)^i\ket{1}$, $\ket{\xi_j}:=\ket{1}+(-1)^j\ket{2}$. We ignore the normalization of the quantum states for simplicity. This t-OPB possesses an elegant geometric representation on the $3\times3\times 3$ block-cube (see Fig.\ref{fig1}). The state $\ket{\psi}_{kkk}$ appears in the body-diagonal.
Each of the sub-blocks  $\mathcal{B}_l$, with $l\in\{1,\cdots,6\}$, consists of four adjacent basic blocks and resides on one of the six faces of the cube. All the sub-blocks and body-diagonal blocks are mutually non overlapping which guarantees orthogonality among the states in different blocks. Orthogonality among the states within a sub-block is guaranteed through construction. We will now construct a UPB from the t-OPB (\ref{tCOPB}). For that, let us first introduce a state  $\ket{S}:=(\ket{0}+\ket{1}+\ket{2})^{\otimes 3}$, which spreads all over the cube in Fig.\ref{fig1}. 
\begin{figure}[t!]
	\begin{center}
		\includegraphics[width=8cm,height=3.5cm]{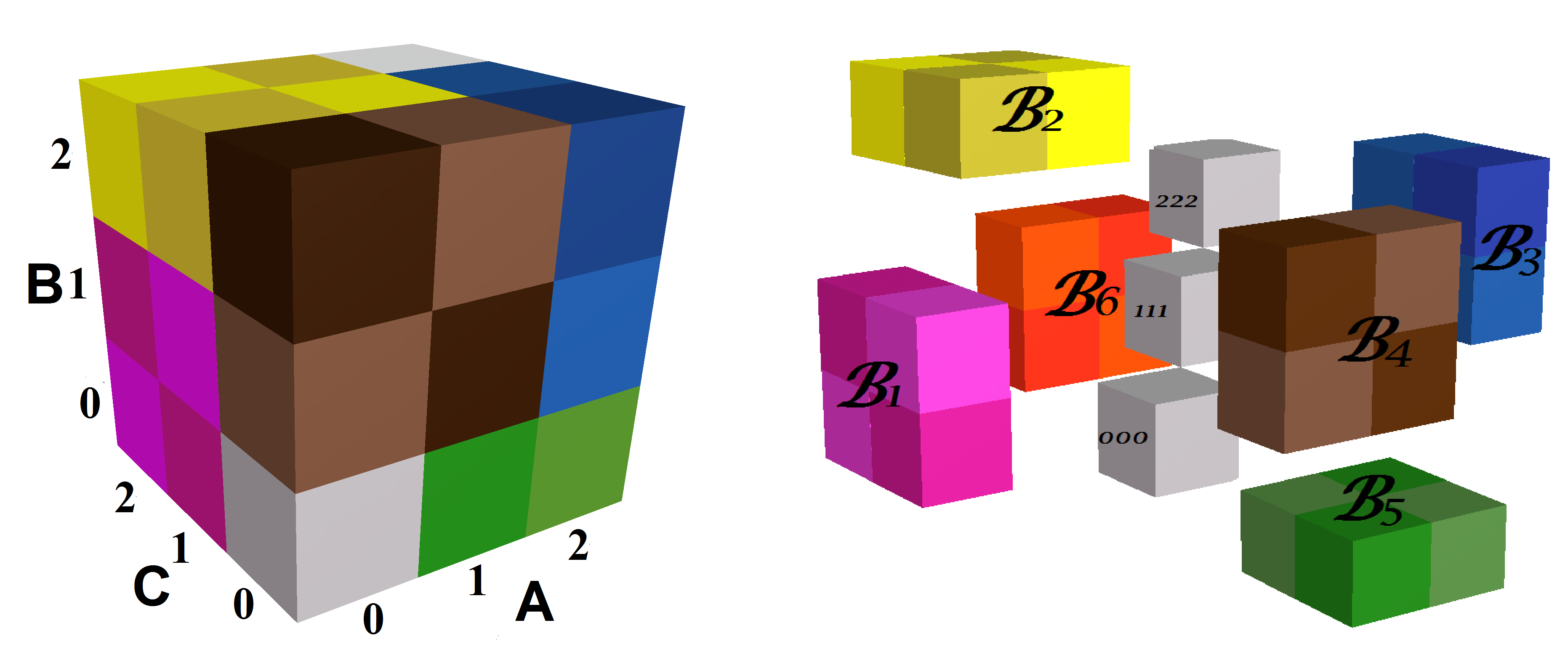}
		\caption{(Color online) The block structure for the t-OPB (\ref{tCOPB}). An arbitrary sub-block made of $n$ basic blocks can contain at most $n$ mutually orthogonal states. Each sub-block $\mathcal{B}_l$ is made of $4$ basic blocks, $\forall~l\in\{1\cdots6\}$. In the right figure all the six sub-blocks and the three body-diagonal blocks are shown.}
		\label{fig1}
	\end{center}
\end{figure}
\begin{proposition}\label{prop}
The set of $19$ states, $\mathcal{U}^{S}_{PB}:=\left\{\bigcup_{l=1}^6\left\{\mathcal{B}_l\setminus\ket{\psi(0,0)}_l\right\}\bigcup\ket{S}\right\}$ is a UPB in $(\mathbb{C}^3)^{\otimes 3}$.
\end{proposition}
\begin{proof}
The argument simply follows from the block structure (Fig.\ref{fig1}). The states in $\mathcal{B}_0$ and the state $\ket{\psi(0,0)}_l\in\mathcal{B}_l,~\forall~l\in\{1,\cdots,6\}$ are $\not\perp$ to $\ket{S}$. Linear combinations of these $9$ missing states can be $\perp$ to $\ket{S}$. But, any such linear combination is always entangled at least in one of three possible bipartitions, which is evident from the block structure. Therefore the orthogonal complement of the subspace spanned by the set of vectors $\mathcal{U}^{S}_{PB}$ contains no product state. 
\end{proof}
The above UPB is symmetric under cyclic permutation of the parties and leads to a $8$ dimensional CES which we denote by $\mathcal{CE}(8)$. 
By relaxing the condition of full separability and allowing biseparable states one can extend the set $\mathcal{U}^{S}_{PB}$. For example, the state $\ket{\psi}^-_{12}:=\ket{\psi(0,0)}_1-\ket{\psi(0,0)}_2$ is orthogonal to all the states in $\mathcal{U}^{S}_{PB}$. This state is separable in AB|C cut, but entangled in other two cuts. Similarly, $\ket{\psi}^-_{34}$ and $\ket{\psi}^-_{56}$ are two other states that are separable in AC|B and A|BC cuts, respectively. The states $\ket{\psi}^-_{12}$, $\ket{\psi}^-_{34}$, and $\ket{\psi}^-_{56}$ are also mutually orthogonal to each other. These three states along with the states in $\mathcal{U}^{S}_{PB}$ form a UBB as shown in the following proposition. 
\begin{proposition}\label{prop1}
	The set of $22$ states, $\mathcal{U}^S_{BB}:=\left\{\mathcal{U}^{S}_{PB}\bigcup\left\{\ket{\psi}^-_{12},\ket{\psi}^-_{34},\ket{\psi}^-_{56}\right\}\right\}$ is a UBB in $(\mathbb{C}^3)^{\otimes 3}$.
\end{proposition}
\begin{proof}
The state $\ket{\psi}^-_{12}$ spreads over the sub-blocks $\mathcal{B}_1$ and $\mathcal{B}_2$ (see Fig.1 of the manuscript). These two sub-blocks are made of eight basic blocks and hence they together can contain at most eight mutually orthogonal states. Seven of which have already been kept there (six fully product states and the biseparable state $\ket{\psi}^-_{12}$). The only other state $\perp$ to these seven states is $\ket{\psi}^+_{12}:=\ket{\psi(0,0)}_1+\ket{\psi(0,0)}_2$. However this state is $\not\perp$ to $\ket{S}$ and hence cannot be appended to the set $\mathcal{U}^S_{BB}$. Same is true for states $\ket{\psi}^+_{34}$ and $\ket{\psi}^+_{56}$. Furthermore, any linear combination of these three states and the states in $\mathcal{B}_0$ is either $\not\perp$ to $\ket{S}$ or genuinely entangled, and hence cannot be appended to $\mathcal{U}^S_{BB}$. This completes the proof.
\end{proof}

The $5$-dimensional subspace $\mathcal{GE}(5)$, complementary to $\mathcal{U}^S_{BB}$, is genuinely entangled. For our case, let us proceed to prove that the subspace $\mathcal{GE}(5)$ is $1$-distillable across all bipartitions. To do so, we first state the following lemma. 
\begin{lemma}\label{lemma}
Consider an $n$-dimensional subspace $\mathcal{S}_{\alpha\beta}$ of a bipartite Hilbert space $\mathbb{C}^{d_{\alpha}}\otimes\mathbb{C}^{d_{\beta}}$. If the projector $\mathbb{P}_{\alpha\beta}$ on $\mathcal{S}_{\alpha\beta}$ satisfies the condition $\mathcal{R}(\mathbb{P}_{\alpha\beta})<\max [\mathcal{R}(\mathbb{P}_{\alpha}),\mathcal{R}(\mathbb{P}_{\beta})]$, then all the rank-$n$ states supported on $\mathcal{S}_{\alpha\beta}$ are $1$-distillable; where $\mathcal{R}(*)$ denotes rank of a matrix and $\mathbb{P}_{\alpha(\beta)}:=\mbox{Tr}_{\beta(\alpha)}[\mathbb{P}_{\alpha\beta}]$. 
\end{lemma}
\begin{proof}
It is known that, an arbitrary bipartite state $\rho_{\alpha\beta}$ is distillable (indeed $1$-distillable \cite{Self1}) if $\mathcal{R}(\rho_{\alpha\beta})<\max [\mathcal{R}(\rho_{\alpha}),\mathcal{R}(\rho_{\beta})]$ \cite{Horodecki03}. Therefore the $n$-rank state proportional to $\mathbb{P}_{\alpha\beta}$ is $1$-distillable. An arbitrary rank-$n$ density operator $\sigma_{\alpha\beta}$ supported on $\mathcal{S}_{\alpha\beta}$ allows a decomposition of the form $\sigma_{\alpha\beta}=p~\mathbb{P}_{\alpha\beta}+(1-p)\chi_{\alpha\beta}$, where $p\in (0,1]$ and $\chi_{\alpha\beta}$ is some other density matrix (of rank $\le n$) supported on $\mathcal{S}_{\alpha\beta}$. Such a decomposition guarantees that $\mathcal{R}(\sigma_{\alpha(\beta)})$ cannot be less than $\mathcal{R}(\mathbb{P}_{\alpha(\beta)})$, where $\sigma_{\alpha(\beta)}:=\mbox{Tr}_{\beta(\alpha)}[\sigma_{\alpha\beta}]$. This completes the proof.
\end{proof}
Consider an $n$ dimensional subspace of a bipartite Hilbert space. Consider arbitrary projector (of rank $k\le n$) supported on that subspace. If all such projectors satisfy the criterion of Lemma \ref{lemma} then the subspace indeed turns out to be distillable. Here our interest is to show distillability of the subspace $\mathcal{GE}(5)$. 
\begin{theorem}\label{theorem2}
	The subspace $\mathcal{GE}(5)$ is distillable across every bipartite cut.
\end{theorem}
\begin{proof}
Distillability of all rank-$5$ states (across every bipartitions) supported on $\mathcal{GE}(5)$ simply follows from the rank of the reduced bimarginals of the full projector on $\mathcal{GE}(5)$, which turns out to be $6$ (see Appendix \ref{appen1}). To show the distillability of other states we need to show that all the projectors (of rank $<5$) supported on $\mathcal{GE}(5)$ satisfy the criterion of lemma \ref{lemma}. Since such a projector can be chosen in infinite ways, we thus require an argument that encapsulates all these infinite possibilities. We formalize such an argument using the block structure and the properties of the missing states

Consider an arbitrary projector $\mathbb{P}(n)$ of rank-$n$ acting on $\mathcal{GE}(5)$, where $1\le n\le 5$. The following facts holds true:
\begin{itemize}
	\item[(I)] Any vector belonging in the subspace $\mathcal{GE}(5)$ must be expressed as a linear combination of some states from the set of missing states $\mathcal{M}:=\{\{\ket{kkk}\}_{k=0,1,2}\bigcup\{\ket{\psi}^+_{ij}\}_{i=1,3,5}~|j=i+1\}$.
	\item[(II)] To construct $n$ mutually orthogonal vectors ($1\le n\le 5$) in the subspace $\mathcal{GE}(5)$ at least $(n+1)$ missing states from the set $\mathcal{M}$ are required.
	\item[(III)] Consider arbitrary $k$ number of states $\ket{m}_1,\cdots,\ket{m}_k$ from the set $\mathcal{M}$. For any such choice $\mathcal{R}(\mbox{Tr}_{\alpha}[\sum_{i=1}^k\ket{m}_i\bra{m}])\ge k$, where $\alpha\in\{A,B,C\}$. In other word, while considering mixture of the states from $\mathcal{M}$, each state contributes at least one independent rank in every bi-marginal.
	\item[(IV)] Consider any two arbitrary states $\ket{m'}$ and $\ket{m''}$ from the set $\mathcal{M}$. The state $\ket{\mathcal{L}(m',m'')}\in\mathcal{GE}(5)$ which is obtained through linear combination of $\ket{m'}$ and $\ket{m''}$ has the property $\mathcal{R}(\mbox{Tr}_{\alpha}[\ket{\mathcal{L}(m',m'')}\bra{\mathcal{L}(m',m'')}]):=r_{\alpha}\ge 2$, for $\alpha\in\{A,B,C\}$. Consider now another state $\ket{m}$ from $\mathcal{M}$ other than $\ket{m'}$ and $\ket{m''}$. For any such state $\ket{m}$ $\mathcal{R}(\mbox{Tr}_{\alpha}[\ket{\mathcal{L}(m',m'')}\bra{\mathcal{L}(m',m'')}+\ket{m}\bra{m}])\ge r_{\alpha}+1$, for $\alpha\in\{A,B,C\}$. In other words, the state $\ket{\mathcal{L}(m',m'')}$ contributes at least two independent rank in the bimarginals with respect to the other states in $\mathcal{M}$.
\end{itemize} 
Let us now consider a set $\Lambda$ of $n$ number of states, i.e., $\Lambda:=\{\ket{\mathcal{L}(m',m'')}, \ket{m}_1,\cdots,\ket{m}_{n-1}\}$. Here $\ket{m'}$ and $\ket{m''}$ are two arbitrarily chosen vectors from $\mathcal{M}$ and $\{\ket{m}_i\}_{i=1}^{n-1}$ are again arbitrarily chosen vectors from $\mathcal{M}$ other than $\ket{m'}$ and $\ket{m''}$. The facts (III) and (IV) assure that rank of all the bimarginals of mixture of $n$ states of any such set $\Lambda$ is at least $(n+1)$. Note that apart from $\ket{\mathcal{L}(m',m'')}$ no other state in $\Lambda$ lies in the subspace $\mathcal{GE}(5)$. Construct now a set $\Lambda'$ of $n$ vectors by taking linear combinations of states from $\Lambda$ such that every state in $\Lambda'$ lies in $\mathcal{GE}(5)$. Rank of the bimarginals of the convex mixture of the states in $\Lambda'$ cannot be less than that of the states in $\Lambda$. Therefore bimarginals of an arbitrary projector $\mathbb{P}(n)$ have rank at least $n+1$ which assures distillability of the normalized projectors across every bipartitions.
\end{proof}

The present construction is the first nontrivial example of a GES which is also distillable across every bipartition. Example of nontrivial distillable subspaces even in the bipartite case are known for lower dimensional system only. Clearly a distillable subspace must be a negative partial transpose (NPT) subspace. Therefore, for $\mathbb{C}^{d_1}\otimes{C}^{d_2}$ system, dimension of distillable subspaces is upper bounded by $(d_1-1)(d_2-1)$ \cite{Johnston13}. Since any rank-$4$ bipartite NPT states are distillable \cite{Chen16}, therefore, when the composite system dimension is not more than $9$, the NPT subspace is indeed a distillable subspaces and the explicit construction follows from Ref. \cite{Johnston13}. Though explicit construction of NPT subspaces is known for higher dimensional systems \cite{Johnston13}, but distillability of those subspaces is not immediate. In fact in Ref. \cite{Chen16} the authors have conjectured a bound NPT states of rank-$5$.  
  
While constructing $\mathcal{U}^S_{BB}$, we add three biseparable states symmetrically (one biseparable state in each cut) to the set $\mathcal{U}^{S}_{PB}$. Instead of this, if one keeps adding biseparable states in one particular cut, say AB|C cut, then it results to another UBB as stated in the following proposition. 
\begin{proposition}\label{prop2}
The set of $23$ states $\mathcal{U}^{AB|C}_{BB}:=\left\{\mathcal{U}^{S}_{PB}\bigcup\left\{\ket{\psi}^-_{12},\ket{\psi}^-_{35}\right\}\bigcup\left\{\ket{\psi}^-_{(0)4},\ket{\psi}^-_{(2)6}\right\}\right\}$ is a UBB in $(\mathbb{C}^3)^{\otimes 3}$; where $\ket{\psi}^-_{(k)l}:=4\ket{\psi}_{kkk}-\ket{\psi(0,0)}_l$.   
\end{proposition} 
\begin{proof}
	 We have the set of $23$ mutually orthogonal states,
	\begin{equation}
	\mathcal{U}^{AB|C}_{BB}:=\{\mathcal{U}^{S}_{PB},\ket{\psi}^-_{12},\ket{\psi}^-_{35},\ket{\psi}^-_{(0)4},\ket{\psi}^-_{(2)6}\}.
	\end{equation}
	As already discussed in Proposition 2, the state $\ket{\psi}^-_{12}$ spreads over the blocks $\mathcal{B}_1~\&~\mathcal{B}_2$ and six more states from this two blocks have already been appended in $\mathcal{U}^{AB|C}_{BB}$. The only state residing in blocks $\mathcal{B}_1~\&~\mathcal{B}_2$ and which is orthogonal to all these seven states is $\ket{\psi}^+_{12}$, but this state is not orthogonal to $\ket{S}$. Same is true for the state $\ket{\psi}^+_{35}$ spreading over the blocks $\mathcal{B}_3~\&~\mathcal{B}_5$. Now, consider the state $\ket{\psi}^-_{(0)4}$, which spreads over the block $\mathcal{B}_4$ and the basic block $000$. At most five mutually orthogonal states can be kept in $\mathcal{B}_4~\&~ 000$, together. Already, we have kept $4$ states -- $\ket{\psi}^-_{(0)4}$ and three other fully product states. The only state that can be orthogonal to all theses $4$ states is $\ket{\psi}^+_{(0)4}:=\ket{\psi}_{000}+\ket{\psi(0,0)}_4$, but it is non orthogonal to $\ket{S}$ and thus can not be appended to $\mathcal{U}^{AB|C}_{BB}$. Similar reasoning holds true for the state $\ket{\psi}^+_{(2)6}$. The remaining state $\ket{\psi}_{111}$ is again non orthogonal to $\ket{S}$ and any linear combination of the states $\{\ket{\psi}_{111},\ket{\psi}^+_{12},\ket{\psi}^+_{35},\ket{\psi}^+_{(0)4},\ket{\psi}^+_{(2)6}\}$which is $\perp$ to $\ket{S}$ is genuinely entangled.
	
	The set $\mathcal{U}^{AB|C}_{BB}$ is indeed a UPB in AB|C cut. While viewing it as a UPB in AB|C cut, it has a tile structure as shown in Fig.\ref{figt}.
	\begin{figure}[t!]
		\begin{center}
			\includegraphics[width=8cm,height=4cm]{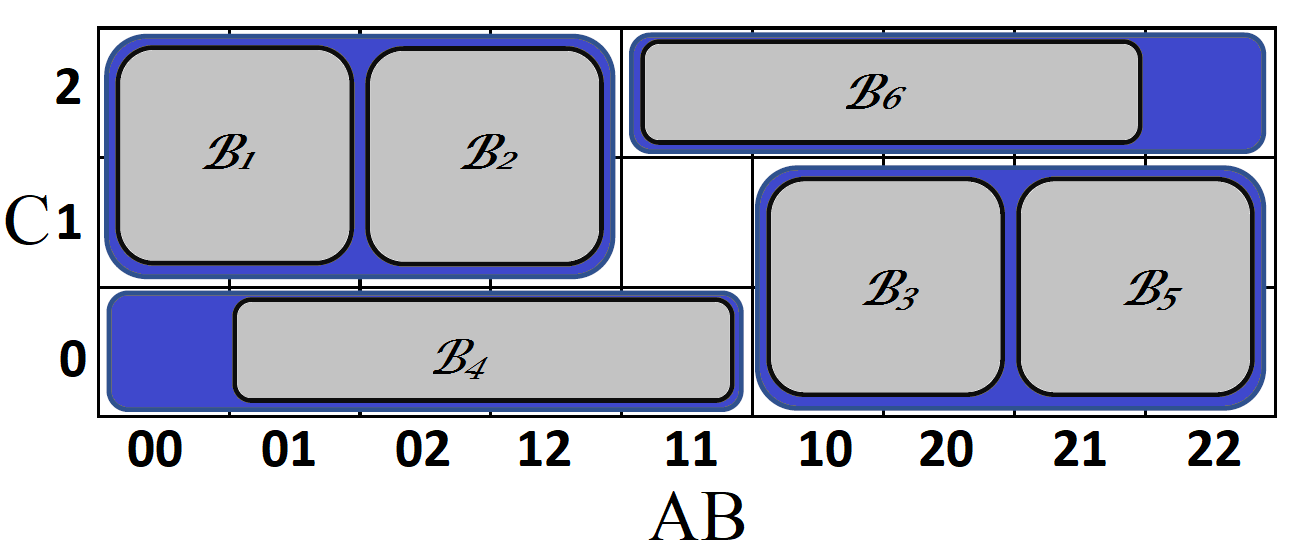}
			\caption[abc]{(Color online) Tile structure of the UPB, $\mathcal{U}^{AB|C}_{BB}$ in AB|C cut. This particular tile structure is a bit similar to that of $\mathbb{C}^3\otimes\mathbb{C}^3$ tile UPB \cite{Bennett99} -- for that look into four blue tiles and the central (white) tile.}
			\label{figt}
		\end{center}
	\end{figure}
	
	In a similar manner one can construct two other UBBs:
	\begin{eqnarray}
	\mathcal{U}^{A|BC}_{BB}&:=&\{\mathcal{U}^{S}_{PB},\ket{\psi}^-_{24},\ket{\psi}^-_{56},\ket{\psi}^-_{(0)1},\ket{\psi}^-_{(2)3}\},\\
	\mathcal{U}^{AC|B}_{BB}&:=&\{\mathcal{U}^{S}_{PB},\ket{\psi}^-_{16},\ket{\psi}^-_{34},\ket{\psi}^-_{(0)5},\ket{\psi}^-_{(2)2}\}
	\end{eqnarray}
	Note that, each state in $\mathcal{U}^{\alpha|\beta}_{BB}\setminus\mathcal{U}^{S}_{PB}$ is $\not\perp$ to at least one of the states in $\mathcal{U}^{\alpha'|\beta'}_{BB}\setminus\mathcal{U}^{S}_{PB}$ , where $\alpha,\alpha'\in\{A,B,C\}$, $\beta,\beta'\in\{BC,AC,AB\}$, and $\alpha\neq\alpha$', $\beta\neq\beta'$.
\end{proof}

Note that $\mathcal{U}^{AB|C}_{BB}$ is indeed a UPB in AB|C cut. Similar constructions are also possible in other two cuts. The remaining $4$ dimensional genuinely entangled subspace $\mathcal{GE}^{AB|C}(4)$ has an important distinct feature than the subspace $\mathcal{GE}(5)$.

\begin{lemma}\label{lemma1}
The three-qutrit density matrix $\rho^{AB|C}_{\mathcal{GE}}(4)$ proportional to the projector on the subspace $\mathcal{GE}^{AB|C}(4)$ is a bound entangled state in AB|C cut, while it is $1$-distillable in other two cuts.
\end{lemma}
\begin{proof}
The state $\rho^{AB|C}_{\mathcal{GE}}(4)$ is bound entangled across AB|C cut follows from the fact that it is a PPT entangled state in this cut. Distillability in other two cuts follows from the marginal rank condition (see Appendix \ref{appen2}).
\end{proof}
Till now we have proposed construction of GESs from UBB. Whether it is possible to construct a GES directly from a UPB such that not even biseparable state can be appended to the set remains an interesting open question. Though such a {\it stronger} UPB is not possible in $(\mathbb{C}^2)^{\otimes m}$ \cite{DiVincenzo00}, its existence in higher dimension will sufficiently lead to a genuinely entangled state which is PPT entangled across every bipartition. However, in the following lemma, we state an interesting property of the present UPB. 
\begin{lemma}\label{lemma2}
The three-qutrit density matrix $\rho(8)$ proportional to the projector on the subspace $\mathcal{CE}(8)$ is bound entangled in every bipartite cut.
\end{lemma}
\begin{proof}
	Since $\mathcal{U}^{S}_{PB}$ is a multipartite UPB, the state $\rho(8)$ is PPT in every cut. However, $\mathcal{U}^{S}_{PB}$ is not a UPB while viewed in a particular bipartition and hence entanglement in that bipartition will not follow directly. But, from the block structure it follows that in a particular cut (say AB|C) there exist only four linearly independent biseparable states -- precisely the four states appended with the set $\mathcal{U}^{S}_{PB}$ to construct the asymmetric UBB in Proposition 3 -- that are orthogonal to the subspace spanned by the set of vectors in $\mathcal{U}^{S}_{PB}$. Thus in the complement subspace we have biseparable state (in AB|C cut) {\it deficit} \cite{Horodecki97,Horodecki98}, which further implies that the state in entangled in that cut. Similar reasoning holds true in other two cuts.
\end{proof}
To the best of our knowledge, all previously known multipartite UPBs \cite{Bennett99,Bennett99(1),Alon01,DiVincenzo03,Parthasarathy04,Hayden06,Niset06,Feng06,Cubitt08,Walgate08,Augusiak11,Childs13,Johnston14,Chen15} are completable in bipartite cut in the sense that allowing separable states (across some cut) one can construct a complete orthogonal bases. Therefore Lemma \ref{lemma2} is a typical feature of the present UPB. Although examples of multipartite states that are PPT across all bipartitions are known \cite{Piani07,Toth09,Huber14,Lockhart18}, the present one is the first example that results from a multipartite UPB. 

Using the structural elegance we generalize the above constructions in $(\mathbb{C}^d)^{\otimes 3}$, with $d\ge 4$. In the next two section we provide the explicit construction for $d=4$ and $d=5$ respectively and then give the generalization for arbitrary dimension. We only provide the constructions of t-OPB, UPB and UBBs.  All the Lemmas, Propositions, and the Theorem proven for $(\mathbb{C}^3)^{\otimes 3}$ also hold true in higher dimensions. 

\section{Constructions in $(\mathbb{C}^4)^{\otimes 3}$}
The twisted OPB (see Fig.\ref{fig2}), in this case, is given by, 
\begin{subequations}\label{tCOPB4}
	\begin{align}
	\mathcal{B}_0:&=\{\ket{\psi}_{kkk}\equiv\ket{k}_A\otimes\ket{k}_B\otimes\ket{k}_C~|~k=0,3\},\\
	\mathcal{B}'_0:&=\{\ket{\psi(l,m,p)}\equiv\ket{\phi_l}_A\otimes\ket{\phi_m}_B\otimes\ket{\phi_p}_C\},\\
	\mathcal{B}_1:&=\{\ket{\psi(i,j)}_1\equiv\ket{0}_A\otimes\ket{\eta_i}_B\otimes\ket{\xi_j}_C\},\\
	\mathcal{B}_2:&=\{\ket{\psi(i,j)}_2\equiv\ket{\eta_i}_A\otimes\ket{3}_B\otimes\ket{\xi_j}_C\},\\
	\mathcal{B}_3:&=\{\ket{\psi(i,j)}_3\equiv\ket{\xi_j}_A\otimes\ket{0}_B\otimes\ket{\eta_i}_C\},\\
	\mathcal{B}_4:&=\{\ket{\psi(i,j)}_4\equiv\ket{\xi_j}_A\otimes\ket{\eta_i}_B\otimes\ket{3}_C\},\\
	\mathcal{B}_5:&=\{\ket{\psi(i,j)}_5\equiv\ket{2}_A\otimes\ket{\xi_j}_B\otimes\ket{\eta_i}_C\},\\
	\mathcal{B}_6:&=\{\ket{\psi(i,j)}_6\equiv\ket{\eta_i}_A\otimes\ket{\xi_j}_B\otimes\ket{0}_C\},
	\end{align}
\end{subequations} 
where $l,m,p\in\{0,1\}$, $\ket{\phi_0}:=\ket{1}+\ket{2}$, $\ket{\phi_1}:=\ket{1}-\ket{2}$; $i,j\in\{0,1,2\}$, $\ket{\eta_0}:=\ket{0}+\ket{1}+\ket{2}$ and $\ket{\eta_1},~\ket{\eta_2}$ are the linear
combination of $\{\ket{0},\ket{1},\ket{2}\}$ such that the set of vectors $\{\ket{\eta_i}\}_{i=1}^3$ are pairwise orthogonal, $\ket{\xi_0}:=\ket{1}+\ket{2}+\ket{3}$ and $\ket{\xi_1},~\ket{\xi_2}$ are the linear
combination of $\{\ket{1},\ket{2},\ket{3}\}$ such that $\{\ket{\xi_j}\}_{j=1}^3$ are pairwise orthogonal. 
\begin{figure}[t!]
	\begin{center}
		\includegraphics[width=6cm,height=4.5cm]{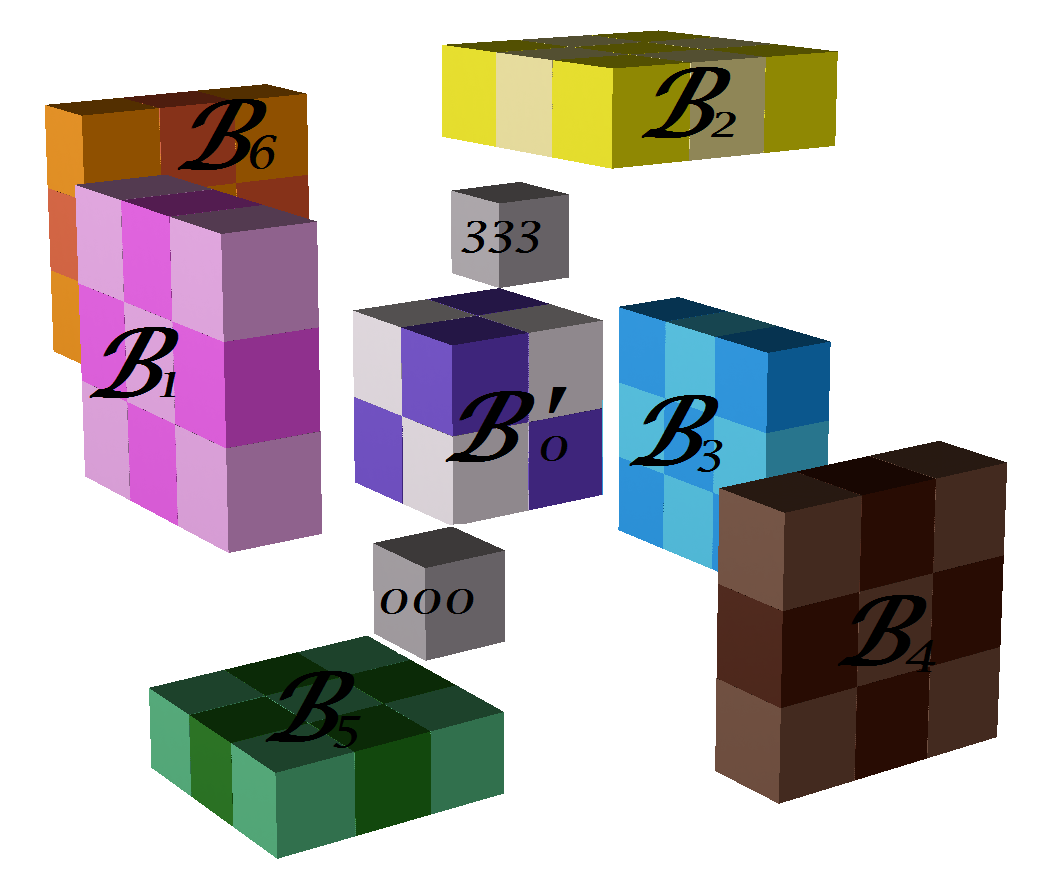}
		\caption{(Color online) The block structure for $(\mathbb{C}^4)^{\otimes 3}$. The states of the t-OPB (\ref{tCOPB4}) are placed in the corresponding Blocks.}
		\label{fig2}
	\end{center}
\end{figure}

In the following we provide the constructions of UPB, symmetric and asymmetric UBBs. Proof follows from the corresponding block structure and the similar kind of arguments discussed for $(\mathbb{C}^3)^{\otimes 3}$.  

{\bf UPB}: In this case, the stopper state is given by, $\ket{S}:=\ket{0+1+2+3}^{\otimes3}$, where $\ket{0+1+2+3}:=\ket{0}+\ket{1}+\ket{2}+\ket{3}$. The state $\ket{S}$ is not orthogonal to the,  states in the block $\mathcal{B}_0$, the state $\ket{\psi(0,0,0)}\in\mathcal{B}'_0$, and the states $\ket{\psi(0,0)}_l\in\mathcal{B}_l$, $\forall~l\in\{1,\cdots,6\}$ of the $\mathbb{C}^4\otimes\mathbb{C}^4\otimes\mathbb{C}^4$ t-OPB (\ref{tCOPB4}). Thus the UPB is given by, $\mathcal{U}^{[4]}_{PB}:=\left\{\cup_{l=1}^6\left\{\mathcal{B}_l\setminus\ket{\psi(0,0)}_l\right\}\cup\left\{\mathcal{B}'_0\setminus\ket{\psi(0,0,0)}\right\}\cup\ket{S}\right\}$.
Superscript (in square brace) denotes the local dimension. In this case the cardinality of $\mathcal{U}^{[4]}_{PB}$ is $56$ leading to a $8$ dimensional CE subspace. 

{\bf UBB (symmetric)}: Appending biseparable states, symmetrically in three cuts, along with $\mathcal{U}^{[4]}_{PB}$  we obtain the symmetric UBB $
\mathcal{U}^{S[4]}_{BB}:=\left\{\mathcal{U}^{[4]}_{PB}\cup\ket{\psi}^-_{12}\cup\ket{\psi}^-_{34}\cup\ket{\psi}^-_{56}\right\}$,
where $\ket{\psi}^-_{12}:=\ket{\psi(0,0)}_1-\ket{\psi(0,0)}_2$; and $\ket{\psi}^-_{34},~\ket{\psi}^-_{56}$ are defined analogously. The resulting genuinely entangled subspace turns out to be $5$ dimensional. The set of missing states are $\mathcal{M}:=\{\{\ket{kkk}\}_{k=0,3}\bigcup\ket{\psi(0,0,0)}\bigcup\{\ket{\psi}^+_{i(i+1)}\}_{i=1,3,5}\}$. Since the missing states posses the same features as stated in the proof of Theorem \ref{theorem2} so the theorem also holds true for this construction. Similar holds true for the higher dimensional constructions.

{\bf UBB (asymmetric)}: The asymmetric UBB, while biseparable states are appended in AB|C cut, is given by $
\mathcal{U}^{AB|C[4]}_{BB}:=\left\{\mathcal{U}^{[4]}_{PB},\ket{\psi}^-_{12},\ket{\psi}^-_{35},\ket{\psi}^-_{(0)4},\ket{\psi}^-_{(3)6}\right\};$
where $\ket{\psi}^-_{(0)4}:=9\ket{\psi}_{000}-\ket{\psi(0,0)}_4$, $\ket{\psi}^-_{(3)6}:=9\ket{\psi}_{333}-\ket{\psi(0,0)}_6$.

\section{Constructions in $(\mathbb{C}^5)^{\otimes 3}$}
\begin{figure}[b!]
	\begin{center}
		\includegraphics[width=6cm,height=5.5cm]{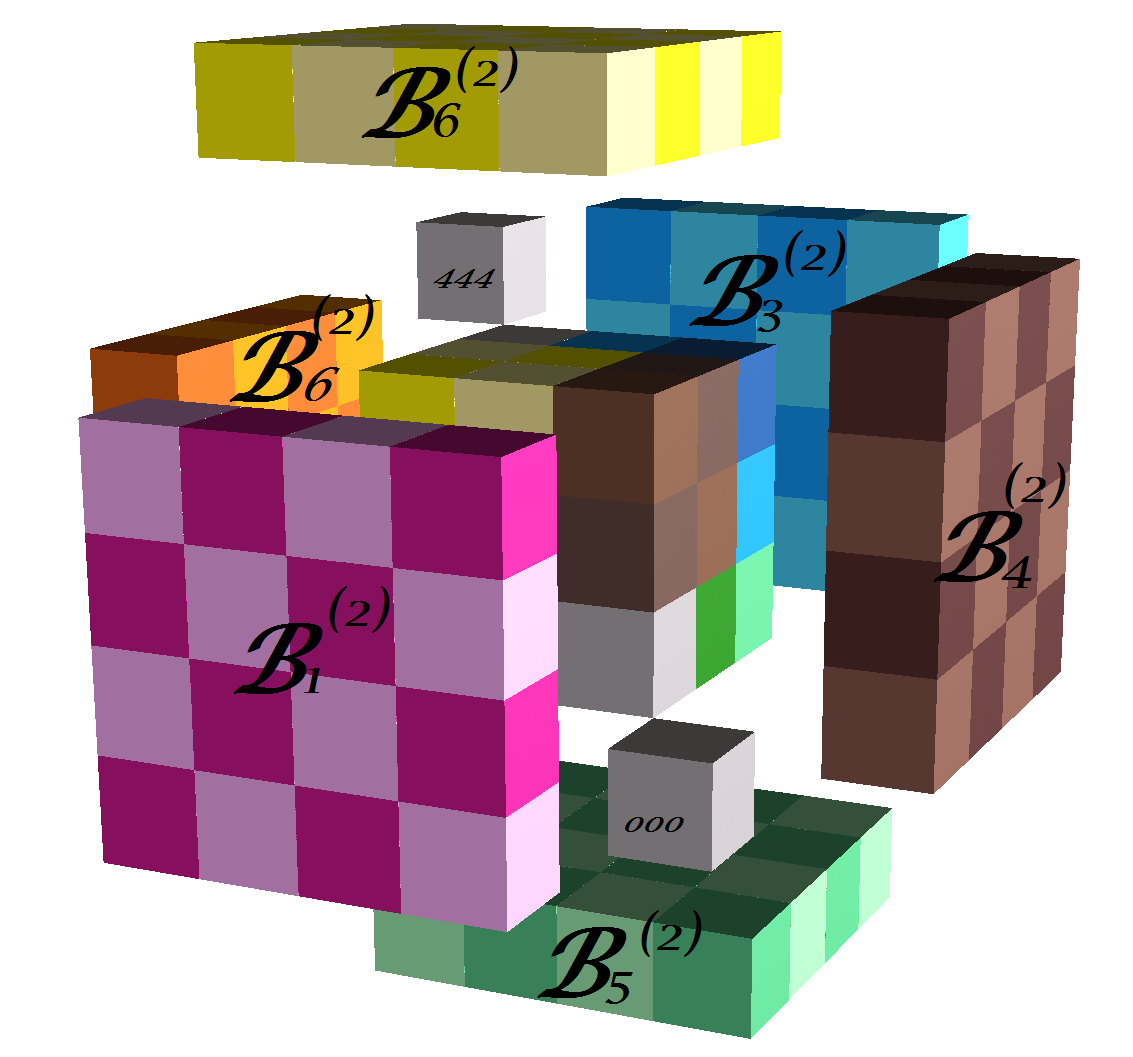}
		\caption{(Color online) The block structure for $(\mathbb{C}^5)^{\otimes 3}$. The states of the t-OPB (\ref{tCOPB5}) are placed in the corresponding Blocks.}
		\label{fig4}
	\end{center}
\end{figure}
Consider the computational bases $\{|i\rangle_A\otimes|j\rangle_B\otimes|k\rangle_C~|~i,j,k=0,1,2,3,4\}$ for a tripartite Hilbert space $\mathbb{C}^5\otimes\mathbb{C}^5\otimes\mathbb{C}^5$. for the same Hilbert space let us consider the following {\it twisted}-OPB arranged in the {\it bolck-cube} of size $5\times 5\times 5$ (see Fig.\ref{fig4}), 
\begin{subequations}\label{tCOPB5}
	\begin{align}
	\mathcal{B}_0:=\{\ket{\psi}_{kkk}\equiv\ket{k}_A\otimes\ket{k}_B\otimes\ket{k}_C|k\in\{0,\cdots,4\}\},\\
	\begin{cases}
	\mathcal{B}^{(1)}_1:&=\{\ket{\psi(i_1,j_1)}_{1}^{(1)}\equiv\ket{1}_A\otimes\ket{\eta_{i_1}}_B\otimes\ket{\xi_{j_1}}_C\},\\
	\mathcal{B}^{(1)}_2:&=\{\ket{\psi(i_1,j_1)}_{2}^{(1)}\equiv\ket{\eta_{i_1}}_A\otimes\ket{3}_B\otimes\ket{\xi_{j_1}}_C\},\\
	\mathcal{B}^{(1)}_3:&=\{\ket{\psi(i_1,j_1)}_{3}^{(1)}\equiv\ket{3}_A\otimes\ket{\xi_{j_1}}_B\otimes\ket{\eta_{i_1}}_C\},\\
	\mathcal{B}^{(1)}_4:&=\{\ket{\psi(i_1,j_1)}_{4}^{(1)}\equiv\ket{\eta_{i_1}}_A\otimes\ket{\xi_{j_1}}_B\otimes\ket{1}_C\},\\
	\mathcal{B}^{(1)}_5:&=\{\ket{\psi(i_1,j_1)}_{5}^{(1)}\equiv\ket{\xi_{j_1}}_A\otimes\ket{1}_B\otimes\ket{\eta_{i_1}}_C\},\\
	\mathcal{B}^{(1)}_6:&=\{\ket{\psi(i_1,j_1)}_{6}^{(1)}\equiv\ket{\xi_{j_1}}_A\otimes\ket{\eta_{i_1}}_B\otimes\ket{3}_C\},\\
	\end{cases}\\
	\begin{cases}
	\mathcal{B}^{(2)}_1:&=\{\ket{\psi(i_2,j_2)}_{1}^{(2)}\equiv\ket{0}_A\otimes\ket{\eta_{i_2}}_B\otimes\ket{\xi_{j_2}}_C\},\\
	\mathcal{B}^{(2)}_2:&=\{\ket{\psi(i_2,j_2)}_{2}^{(2)}\equiv\ket{\eta_{i_2}}_A\otimes\ket{4}_B\otimes\ket{\xi_{j_2}}_C\},\\
	\mathcal{B}^{(2)}_3:&=\{\ket{\psi(i_2,j_2)}_{3}^{(2)}\equiv\ket{4}_A\otimes\ket{\xi_{j_2}}_B\otimes\ket{\eta_{i_2}}_C\},\\
	\mathcal{B}^{(2)}_4:&=\{\ket{\psi(i_2,j_2)}_{4}^{(2)}\equiv\ket{\eta_{i_2}}_A\otimes\ket{\xi_{j_2}}_B\otimes\ket{0}_C\},\\
	\mathcal{B}^{(2)}_5:&=\{\ket{\psi(i_2,j_2)}_{5}^{(2)}\equiv\ket{\xi_{j_2}}_A\otimes\ket{0}_B\otimes\ket{\eta_{i_2}}_C\},\\
	\mathcal{B}^{(2)}_6:&=\{\ket{\psi(i_2,j_2)}_{6}^{(2)}\equiv\ket{\xi_{j_2}}_A\otimes\ket{\eta_{i_2}}_B\otimes\ket{4}_C\},
	\end{cases}
	\end{align}
\end{subequations}
where $i_1,j_2\in\{0,1\}$, and $\ket{\eta_{i_1}},~\ket{\xi_{j_1}}$ are analogous to $\ket{\eta_{i}},~\ket{\xi_{j}}$ of Eq.(\ref{tCOPB}) with $0$ and $1$ replaced by $1$ and $2$, respectively; 
$i_2,j_2\in\{0,1,2,3\}$, $\ket{\eta_{i_2=0}}:=\ket{0+1+2+3}$ and other $\ket{\eta_{i_2}}$ are the linear
combination of $\{\ket{l}\}_{l=0}^3$ with further requirement that the set of vectors $\{\ket{\eta_{i_2}}\}_{i_2=0}^3$ are pairwise orthogonal, $\ket{\xi_{j_2=0}}:=\ket{1+2+3+4}$ and other $\ket{\xi_{j_2}}$'s are the linear
combination of $\{\ket{l}\}_{l=1}^4$ and the set of vectors $\{\ket{\xi_{j_2}}\}_{j_2=0}^3$ are pairwise orthogonal.

Here the states in $\mathcal{B}_0$ are kept in the body-diagonal basic blocks. The facial blocks appears in two layers -- the superscript (within parentheses) on $\mathcal{B}$ in Eqs.(\ref{tCOPB5}) denotes the layer index, $(l)$, with $l=1$ denoting the inner layer and $l=2$ the outer. Each facial block in inner and outer layer consists of $4$ and $16$ basic blocks respectively. 

{\bf UPB}: In the case  the UPB turns out to be $\mathcal{U}^{[5]}_{PB}:=\left\{\bigcup_{p=1,l=1}^{6,2}\left\{\mathcal{B}^{(l)}_p\setminus\ket{\psi(0,0)}^{(l)}_p\right\}\bigcup\ket{S}\right\}$, where $\ket{S}:=\ket{0+1+2+3+4}^{\otimes 3}$. Here, $\mathcal{C}(\mathcal{U}^{[5]}_{PB})=109$ and the dimension of the CE subspace is $16$; $\mathcal{C}$ denotes cardinality.

{\bf UBB {symmetric}}: Two (one from each layer) biseparable states are added in each cut along with the set $\mathcal{U}^{[5]}_{PB}$ to obtain the symmetric UBB $\mathcal{U}^{S[5]}_{BB}:=\left\{\mathcal{U}^{[5]}_{PB}\bigcup_{l=1}^2\left\{\ket{\psi}^{(l)-}_{12},\ket{\psi}^{(l)-}_{34},\ket{\psi}^{(l)-}_{56}\right\}\right\}$; where $\ket{\psi}^{(l)-}_{ij}:=\ket{\psi(0,0)}^{(l)}_i-\ket{\psi(0,0)}^{(l)}_j$. Dimension of the GES is $10$ and a orthonormal basis for this subspace can be constructed from the linear combination of  missing states $\mathcal{M}:=\left\{\left\{\ket{kkk}\right\}_{k=0,...,4} \bigcup_{l=1}^2\left\{\ket{\psi}^{l+}_{ij}\right\}_{i=1,3,5}\right\}$, where $j=i+1$. 

Generalization of the construction for arbitrary local dimension is presented in Appendix \ref{appen3}. 


\section{Discussions} 
Understanding multipartite entanglement is quite demanding as quantum advantages in several real-world tasks can only ever be achieved by exploiting multipartite entanglement as a resource \cite{Hillery99,Raussendorf01,Komer14,Pirker18}; and it also has implications in other branches of physics \cite{Amico08,Facchi09,Rangamani17}. The notion of UBB studied here is important as it sufficiently leads to a subspace containing only genuinely entangled states. Our symmetric UBB leads to a subspace that is not only a GES but also bidistillable, which makes it operationally more useful. Here the question arises whether construction of a multipartite subspace is possible which is not only bidistillable but also multipartite distillable, i.e., any kind of pure genuinely entangled state can be distilled from every state supported on that subspace under local operation and classical communications (LOCC). In particular, subspace of a tripartite Hilbert space will be operationally most useful if sharing identical copies of an arbitrary state $\rho_{ABC}$ supported on that subspace one can distill two-qubit maximally entangled states (e.g. singlet state) between any two parties, i.e., $\rho_{ABC}^{\otimes m}\xrightarrow[]{LOCC}\ket{\psi^-}_{AB}^{n_{1}}\otimes\ket{\psi^-}_{BC}^{n_{2}}\otimes\ket{\psi^-}_{CA}^{n_{3}}$. It will be interesting to check whether $\mathcal{GE}(5)$ is of that kind. Our asymmetric UBB (Proposition \ref{prop2}) indeed turns out to be a UPB in a bipartition. Though such a UBB is not possible in $(\mathbb{C}^2)^{\otimes 3}$ \cite{DiVincenzo00}, but possibility of constructing a UBB in $(\mathbb{C}^2)^{\otimes 3}$ is not ruled out in general; and we leave this here as an open question for future research.

Our study raises a number of other interesting questions. For the present construction the proof that the given set of states forms a UBB follows from the elegant geometric structure of block cube. However, it is not the case in general and therefore demands more efficient method to determine whether an arbitrary given set of orthogonal vectors forms a UBB or not. Furthermore, for a given multipartite Hilbert space, it is also important to know the maximum achievable dimension of GES which is bidistillable. An upper bound follows from the result of Ref. \cite{Johnston13} which turns out to be $(d^3-d^2-d+1)$ for the Hilbert space $(\mathbb{C}^d)^{\otimes 3}$. Finally, the notion of UPB is deeply linked with the notion called `quantum nonlocality without entanglement' \cite{Bennett99(1)}. It will be quite interesting to explore such new concepts that result from the UBB introduced here. For example in \cite{Halder18(1)}, the authors introduced a new concept called `strong quantum nonlocality without entanglement'. While we have checked that the present UPB does not exhibit such phenomena, it will be interesting to study whether a UBB can exhibit that. Generalizing the present construction for more number of parties is also an important direction to explore \cite{self}.  

\begin{acknowledgments}
	The authors gratefully acknowledge useful discussions with Guruprasad Kar, Sibasish ghosh, and Somshubhro Bandyopadhyay. SA would like to acknowledge discussions with Prasanta K. Panigrahi. Private communication with Otfried Gühne is gratefully acknowledged. MB acknowledges support through an INSPIRE-faculty position at S. N. Bose National Centre for Basic Sciences, by the Department of Science and Technology, Government of India.
\end{acknowledgments}  

\appendix
\begin{widetext}
\section{Normalized projector on $\mathcal{GE}(5)$}\label{appen1}
The state $\rho^{S}_{\mathcal{GE}}(5)$ proportional to the projector on subspace $\mathcal{GE}(5)$ is given by,   
\begin{eqnarray}
\rho^{S}_{\mathcal{GE}}(5):=\frac{1}{5}\left(\mathbb{I}_3\otimes\mathbb{I}_3\otimes\mathbb{I}_3-\sum_{\ket{\psi}\in\mathcal{U}^{S}_{BB}}\ket{\tilde{\psi}}\bra{\tilde{\psi}}\right).
\end{eqnarray}
Here $\ket{\tilde{\psi}}$ the normalized state proportional to $\ket{\psi}$. Since the construction is party symmetric, all the two party reduced states $\rho_{\beta}:=\mbox{Tr}_{\alpha}[\rho^{S}_{\mathcal{GE}}(5)]$, with $\beta\in\{BC,CA,AB\}$ and $\alpha\in\{A,B,C\}$, respectively, are identical and the corresponding density matrix takes the following form:
\begin{equation}
\rho_{\beta}
=\frac{1}{360}\left(
\begin{array}{ccccccccc}
82 & 10 & 10 & -8 & -8 & 10 & -8 & -8 & -8 \\
10 & 19 & 19 & -8 & 1 & 19 & -8 & 1 & 1
\\
10 & 19 & 19 & -8 & 1 & 19 & -8 & 1 & 1
\\
-8 & -8 & -8 & 19 & 1 & -8 & 19 & 1 & -8
\\
-8 & 1 & 1 & 1 & 82 & 1 & 1 & 10 & 1 \\
10 & 19 & 19 & -8 & 1 & 19 & -8 & 1 & 1
\\
-8 & -8 & -8 & 19 & 1 & -8 & 19 & 1 & -8
\\
-8 & 1 & 1 & 1 & 10 & 1 & 1 & 19 & 10 \\
-8 & 1 & 1 & -8 & 1 & 1 & -8 & 10 & 82 \\
\end{array}
\right).
\end{equation}
Rank of the state $\rho_{\beta}$ is $6$.

\section{Normalized projector on $\mathcal{GE}(4)$}\label{appen2}
The state $\rho^{AB|C}_{\mathcal{GE}}(4)$ proportional to the projector on subspace $\mathcal{GE}(4)$ is given by,   
\begin{eqnarray}
\rho^{AB|C}_{\mathcal{GE}}(4):=\frac{1}{4}\left(\mathbb{I}_3\otimes\mathbb{I}_3\otimes\mathbb{I}_3-\sum_{\ket{\psi}\in\mathcal{U}^{AB|C}_{BB}}\ket{\tilde{\psi}}\bra{\tilde{\psi}}\right).
\end{eqnarray}
Since all $\ket{\psi}\in\mathcal{U}^{AB|C}_{BB}$ are separable in AB|C cut, partial transpose in this bipartite cut gives a non-negative operator; in fact in this case the state is invariant under partial transposition in AB|C cut. However, in other two cuts the state is NPT. In the following we note down BC and AC marginals of the state: 
\begin{equation}
\rho_{BC}=\frac{1}{1440}\left(
\begin{array}{ccccccccc}
122 & 50 & -40 & 77 & 5 & -40 & 77 & 5 & -40 \\
50 & 95 & 5 & 5 & 50 & 5 & 5 & 50 & 5 \\
-40 & 5 & 149 & -40 & 5 & 149 & -40 & 5 & 77 \\
77 & 5 & -40 & 149 & 5 & -40 & 149 & 5 & -40 \\
5 & 50 & 5 & 5 & 410 & 5 & 5 & 50 & 5 \\
-40 & 5 & 149 & -40 & 5 & 149 & -40 & 5 & 77 \\
77 & 5 & -40 & 149 & 5 & -40 & 149 & 5 & -40 \\
5 & 50 & 5 & 5 & 50 & 5 & 5 & 95 & 50 \\
-40 & 5 & 77 & -40 & 5 & 77 & -40 & 50 & 122 \\
\end{array}
\right);
\end{equation} 
\begin{equation}
\rho_{AC}=\frac{1}{1440}\left(
\begin{array}{ccccccccc}
176 & -40 & -40 & 104 & -40 & -40 & -40 & -40 & -40 \\
-40 & 95 & 95 & -40 & 5 & 5 & -40 & -40 & -40 \\
-40 & 95 & 95 & -40 & 5 & 5 & -40 & -40 & -40 \\
104 & -40 & -40 & 149 & 5 & -40 & 5 & 5 & -40 \\
-40 & 5 & 5 & 5 & 410 & 5 & 5 & 5 & -40 \\
-40 & 5 & 5 & -40 & 5 & 149 & -40 & -40 & 104 \\
-40 & -40 & -40 & 5 & 5 & -40 & 95 & 95 & -40 \\
-40 & -40 & -40 & 5 & 5 & -40 & 95 & 95 & -40 \\
-40 & -40 & -40 & -40 & -40 & 104 & -40 & -40 & 176 \\
\end{array}
\right).
\end{equation}
It turns out that $\mathcal{R}(\rho_{BC})=\mathcal{R}(\rho_{AC})=7$, which further implies that the state $\rho^{AB|C}_{\mathcal{GE}}(4)$ is $1$-distillable in A|BC and AC|B cuts.

\section{Constructions in $(\mathbb{C}^d)^{\otimes 3}$}\label{appen3}

{\bf Odd dimension(s)}: First we consider odd dimensional case, i.e., $d=(2n+1)$, $n$ being positive integers. It is clear from the construction $d=5$, for arbitrary $d$, we have $(d-1)/2$ layers of facial blocks. The t-OPB is given by,   
\begin{subequations}\label{tCOPBd}
	\begin{align}\nonumber
	\mathcal{B}_0:=\{\ket{\psi}_{kkk}\equiv\ket{k}_A\otimes\ket{k}_B\otimes\ket{k}_C|k\in\{0,\cdots,(d-1)\}\},\nonumber\\
	\begin{cases}\nonumber
	...&~~~~~~~~~~~~~~~~~~~~~~~~~~~~~~~~~~~~~~~~~~~~~~~~~~~~~~~~~~~~~~~~~~~\\
	...&~~~~~~~~~~~~~~~~~~~~~~~~~~~~~~~~~~~~~~~~~~~~~~~~~~~~~~~~~~~~~~~~~~~\\
	\end{cases}\\
	...~~~~~~~~~~~~~~~~~~~~~~~~~~~~~~~~~~~~~~~~~~~~~~~~~~~~~~~~~~~~~~~~~~~\nonumber\\
	...~~~~~~~~~~~~~~~~~~~~~~~~~~~~~~~~~~~~~~~~~~~~~~~~~~~~~~~~~~~~~~~~~~~\nonumber\\
	\begin{cases}\nonumber
	\mathcal{B}^{(l)}_1:&=\{\ket{\psi(i_l,j_l)}_{1}^{(1)}\equiv\ket{\frac{(d-1)}{2}-l}_A\otimes\ket{\eta_{i_l}}_B\otimes\ket{\xi_{j_l}}_C\},\\
	\mathcal{B}^{(l)}_2:&=\{\ket{\psi(i_l,j_l)}_{2}^{(1)}\equiv\ket{\eta_{i_l}}_A\otimes\ket{\frac{(d-1)}{2}+l}_B\otimes\ket{\xi_{j_l}}_C\},\\
	\mathcal{B}^{(l)}_3:&=\{\ket{\psi(i_l,j_l)}_{3}^{(1)}\equiv\ket{\frac{(d-1)}{2}+l}_A\otimes\ket{\xi_{j_l}}_B\otimes\ket{\eta_{i_l}}_C\},\\
	\mathcal{B}^{(l)}_4:&=\{\ket{\psi(i_l,j_l)}_{4}^{(1)}\equiv\ket{\eta_{i_l}}_A\otimes\ket{\xi_{j_1}}_B\otimes\ket{\frac{(d-1)}{2}-l}_C\},\\
	\mathcal{B}^{(l)}_5:&=\{\ket{\psi(i_l,j_l)}_{5}^{(1)}\equiv\ket{\xi_{j_l}}_A\otimes\ket{\frac{(d-1)}{2}-l}_B\otimes\ket{\eta_{i_l}}_C\},\\
	\mathcal{B}^{(l)}_6:&=\{\ket{\psi(i_l,j_l)}_{6}^{(1)}\equiv\ket{\xi_{j_l}}_A\otimes\ket{\eta_{i_l}}_B\otimes\ket{\frac{(d-1)}{2}+l}_C\},\\
	\end{cases}\\
	...~~~~~~~~~~~~~~~~~~~~~~~~~~~~~~~~~~~~~~~~~~~~~~~~~~~~~~~~~~~~~~~~~~~\nonumber\\
	...~~~~~~~~~~~~~~~~~~~~~~~~~~~~~~~~~~~~~~~~~~~~~~~~~~~~~~~~~~~~~~~~~~~\nonumber\\
	\end{align}
\end{subequations}  
Here, for $l^{th}$ layer, $i_l,j_l\in\{0,\cdots,2l-1\}$; $\eta_{i_l=0}:=\ket{0+\cdots+(2l-1)}$, and other $\eta_{i_l}$'s are linear combinations of $\{\ket{k}\}_{k=0}^{l-1}$ in such a way that $\{\eta_{i_l}\}_{i_l=0}^{2l-1}$ are mutually orthogonal; $\xi_{j_l=0}:=\ket{1+\cdots+2l}$, and other $\xi_{j_l}$'s are linear combinations of $\{\ket{k}\}_{k=1}^{l}$ in such a way that $\{\xi_{j_l}\}_{j_l=0}^{2l-1}$ are mutually orthogonal.

{\bf UPB}: $
\mathcal{U}^{[d]}_{PB}:=\left\{\bigcup_{p=1,l=1}^{6,(d-1)/2}\left\{\mathcal{B}^{(l)}_p\setminus\ket{\psi(0,0)}^{(l)}_p\right\}\bigcup\ket{S}\right\}$, where $\ket{S}:=\ket{0+\cdots+(d-1)}^{\otimes 3}$.

{\bf UBB {symmetric}}: $\mathcal{U}^{S[d]}_{BB}:=\left\{\mathcal{U}^{[d]}_{PB}\bigcup_{l=1}^{(d-1)/2}\left\{\ket{\psi}^{(l)-}_{12},\ket{\psi}^{(l)-}_{34},\ket{\psi}^{(l)-}_{56}\right\}\right\}$.
Set of missing states are $\mathcal{M}:=\left\{\left\{\ket{kkk}\right\}_{k=0,...,d-1} \bigcup_{l=1}^{\frac{d-1}{2}}\left\{\ket{\psi}^{l+}_{ij}\right\}_{i=1,3,5}\right\}$, where $j=i+1$.

{\bf UBB {asymmetric}}: $ \mathcal{U}^{AB|C[d]}_{BB}:=\left\{\mathcal{U}^{[d]}_{PB}\bigcup_{l=1}^{(d-1)/2}\left\{\ket{\psi}^{(l)-}_{12},\ket{\psi}^{(l)-}_{35}\right\}\bigcup_{r=1}^{(d-1)/2}\left\{\ket{\psi}^{(r)-}_{\left(\frac{d-1}{2}-r\right)4},\ket{\psi}^{(r)-}_{\left(\frac{d-1}{2}+r\right)6}\right\}\right\}$.

{\bf Even dimension(s)}: $d=2n$, $n$ being positive integers. In this case we have $(d/2-1)$ layers of facial blocks, but layered around $\mathbb{C}^4\otimes\mathbb{C}^4\otimes\mathbb{C}^4$ structure. The twisted COPB is given by,
\begin{subequations}\label{tCOPBde}
	\begin{align}\nonumber
	\mathcal{B}_0:=\left\{\ket{\psi}_{kkk}\equiv\ket{k}_A\otimes\ket{k}_B\otimes\ket{k}_C|k\in\left\{0,\cdots,(d-1)\right\}\setminus\frac{d\pm 1}{2}\right\},\nonumber\\
	\mathcal{B}'_0:=\{\ket{\psi(l,m,p)}\equiv\ket{\phi_l}_A\otimes\ket{\phi_m}_B\otimes\ket{\phi_p}_C~|~l,m,p\in\{0,1\},\nonumber\\
	\begin{cases}\nonumber
	...&~~~~~~~~~~~~~~~~~~~~~~~~~~~~~~~~~~~~~~~~~~~~~~~~~~~~~~~~~~~~~~~~~~~\\
	...&~~~~~~~~~~~~~~~~~~~~~~~~~~~~~~~~~~~~~~~~~~~~~~~~~~~~~~~~~~~~~~~~~~~\\
	\end{cases}\\
	...~~~~~~~~~~~~~~~~~~~~~~~~~~~~~~~~~~~~~~~~~~~~~~~~~~~~~~~~~~~~~~~~~~~\nonumber\\
	...~~~~~~~~~~~~~~~~~~~~~~~~~~~~~~~~~~~~~~~~~~~~~~~~~~~~~~~~~~~~~~~~~~~\nonumber\\
	\begin{cases}\nonumber
	\mathcal{B}^{(l)}_1:&=\{\ket{\psi(i_l,j_l)}_{1}^{(1)}\equiv\ket{\frac{d}{2}-1-l}_A\otimes\ket{\eta_{i_l}}_B\otimes\ket{\xi_{j_l}}_C\},\\
	\mathcal{B}^{(l)}_2:&=\{\ket{\psi(i_l,j_l)}_{2}^{(1)}\equiv\ket{\eta_{i_l}}_A\otimes\ket{\frac{d}{2}+l}_B\otimes\ket{\xi_{j_l}}_C\},\\
	\mathcal{B}^{(l)}_3:&=\{\ket{\psi(i_l,j_l)}_{3}^{(1)}\equiv\ket{\frac{d}{2}+l}_A\otimes\ket{\xi_{j_l}}_B\otimes\ket{\eta_{i_l}}_C\},\\
	\mathcal{B}^{(l)}_4:&=\{\ket{\psi(i_l,j_l)}_{4}^{(1)}\equiv\ket{\eta_{i_l}}_A\otimes\ket{\xi_{j_1}}_B\otimes\ket{\frac{d}{2}-1-l}_C\},\\
	\mathcal{B}^{(l)}_5:&=\{\ket{\psi(i_l,j_l)}_{5}^{(1)}\equiv\ket{\xi_{j_l}}_A\otimes\ket{\frac{d}{2}-1-l}_B\otimes\ket{\eta_{i_l}}_C\},\\
	\mathcal{B}^{(l)}_6:&=\{\ket{\psi(i_l,j_l)}_{6}^{(1)}\equiv\ket{\xi_{j_l}}_A\otimes\ket{\eta_{i_l}}_B\otimes\ket{\frac{d}{2}+l}_C\},\\
	\end{cases}\\
	...~~~~~~~~~~~~~~~~~~~~~~~~~~~~~~~~~~~~~~~~~~~~~~~~~~~~~~~~~~~~~~~~~~~\nonumber\\
	...~~~~~~~~~~~~~~~~~~~~~~~~~~~~~~~~~~~~~~~~~~~~~~~~~~~~~~~~~~~~~~~~~~~\nonumber\\
	\end{align}
\end{subequations}
$\ket{\phi_0}:=\ket{d/2-1}+\ket{d/2}$, $\ket{\phi_1}:=\ket{d/2-1}-\ket{d/2}$; for $l^{th}$ layer, $i_l,j_l\in\{0,\cdots,2l\}$; $\eta_{i_l=0}:=\ket{0+\cdots+2l}$, and other $\eta_{i_l}$'s are linear combinations of $\{\ket{k}\}_{k=0}^{2l}$ in such a way that $\{\eta_{i_l}\}_{i_l=0}^{2l}$ are mutually orthogonal; $\xi_{j_l=0}:=\ket{1+\cdots+(2l+1)}$, and other $\xi_{j_l}$'s are linear combinations of $\{\ket{k}\}_{k=1}^{2l+1}$ in such a way that $\{\xi_{j_l}\}_{j_l=0}^{2l}$ are mutually orthogonal.

{\bf UPB}: $
\mathcal{U}^{[d]}_{PB}:=\left\{\bigcup_{p=1,l=1}^{6,(d/2-1)}\left\{\mathcal{B}^{(l)}_p\setminus\ket{\psi(0,0)}^{(l)}_p\right\}\bigcup\left\{\mathcal{B}'_0\setminus\ket{\psi(0,0,0)}\right\}\bigcup\ket{S}\right\}$, where $\ket{S}:=\ket{0+\cdots+(d-1)}^{\otimes 3}$.

{\bf UBB {symmetric}}: $
\mathcal{U}^{S[d]}_{BB}:=\left\{\mathcal{U}^{[d]}_{PB}\bigcup_{l=1}^{(d/2-1)}\left\{\ket{\psi}^{(l)-}_{12},\ket{\psi}^{(l)-}_{34},\ket{\psi}^{(l)-}_{56}\right\}\right\}$ and the missing states are $\mathcal{M}:=\{\{\ket{kkk}\}_{k\in\{0,..d-1\}\setminus\{\frac{d}{2}-1,\frac{d}{2}+1\}}\bigcup\ket{\psi(0,0,0)}\bigcup_{l=1}^{\frac{d}{2}-1}\{\ket{\psi}^{l+}_{ij}\}_{i=1,3,5}\}$, where $j=i+1$.

{\bf UBB {asymmetric}}: $
\mathcal{U}^{AB|C[d]}_{BB}:=\left\{\mathcal{U}^{[d]}_{PB}\bigcup_{l=1}^{(d/2-1)}\left\{\ket{\psi}^{(l)-}_{12},\ket{\psi}^{(l)-}_{35}\right\}\bigcup_{r=1}^{(d/2-1)}\left\{\ket{\psi}^{(r)-}_{\left(\frac{d}{2}-1-r\right)4},\ket{\psi}^{(r)-}_{\left(\frac{d}{2}+r\right)6}\right\}\right\}$.

\end{widetext}


\begin{thebibliography}{99}
\bibitem{Self1} A bipartite density matrix $\rho$ is called $1$-distillable if there exists a pure bipartite state $\ket{\psi}$ of Schmidt rank $2$ such that $\langle\psi|\rho^{PT}|\psi\rangle<0$, where PT denotes partial transposition. The state will be called $n$-distillable if $\rho^{\otimes n}$ is $1$-distillable.
\end{thebibliography}
\end{document}